\documentclass[11pt]{article}
\usepackage{amsmath}
\usepackage{amssymb}
\usepackage{amsfonts}
\usepackage{amsthm}
\usepackage[round]{natbib}
\usepackage{enumitem} 
\usepackage{graphicx}
\usepackage{mathabx}

\newcommand{\FDR}{\textup{\mbox{FDR}}}

\newcommand{\biggbar}{\hspace{1pt} \bigg| \hspace{1pt}}

\newtheorem{thm}{Theorem}

\newtheorem{lem}{Lemma}

\begin{document}

\title{Dynamic adaptive procedures that control the false discovery rate}

\author{Peter MacDonald and Kun Liang\\
Department of Statistics and Actuarial Science\\
\vspace{3mm}
University of Waterloo\\
Arnold Janssen\\
Heinrich-Heine University D\"usseldorf\\
Mathematical Institute\\
}

\maketitle

\abstract{ In the multiple testing problem with independent tests, the classical linear step-up procedure controls the false discovery rate (FDR) at level $\pi_0\alpha$, where $\pi_0$ is the proportion of true null hypotheses and $\alpha$ is the target FDR level.   Adaptive procedures can improve power by incorporating estimates of $\pi_0$, which typically rely on a tuning parameter.   Fixed adaptive procedures set their tuning parameters before seeing the data and can be shown to control the FDR in finite samples.   We develop theoretical results for dynamic adaptive procedures whose tuning parameters are determined by the data.   We show that, if the tuning parameter is chosen according to a stopping time rule, the corresponding dynamic adaptive procedure controls the FDR in finite samples.   Examples include the recently proposed right-boundary procedure and the widely used lowest-slope procedure, among others.   Simulation results show that the right-boundary procedure is more powerful than other dynamic adaptive procedures under independence and mild dependence conditions.}

\setlength\parindent{0pt}
\setlength{\parskip}{\baselineskip}

\section{Introduction}

Powerful modern computers have introduced large data sets to diverse fields of research, and testing of hundreds or even thousands of hypotheses simultaneously has become commonplace in statistical applications such as genetics, neuroscience, and astronomy.   Since its inception in \citet{BH95}, the false discovery rate (FDR), the expected proportion of false positives, has been widely adopted as an error measure for such large-scale problems.   Much research effort has been made to improve Benjamini and Hochberg's initial method, in particular developing efficient estimators of the FDR that lead to powerful procedures which maintain FDR control. In this paper, we provide the proof of finite sample FDR control for a large class of data-adaptive procedures.   First, we briefly review the literature.

Consider the classical problem of testing $m$ independent simultaneous null hypotheses, of which $m_0$ are true and $m_1 = m - m_0$ are false.   Denote the associated $p$-values by $p_1, p_2, ... ,p_m$ and the ordered $p$-values by $p_{(1)} \leq \cdots \leq p_{(m)}$.   For $t \in [0,1]$, define the following empirical processes \cite{St04}:
\begin{eqnarray*}
V(t) &=& \#\{\text{true null } p_i : p_i \leq t \},\\
S(t) &=& \#\{\text{false null } p_i : p_i \leq t\},\\
R(t) &=& V(t) + S(t).
\end{eqnarray*}
Then the FDR at a $p$-value cut-off $t \in (0,1]$ is defined as
\[
	\FDR(t) =  E \bigg[ \frac{V(t)}{R(t) \vee 1} \bigg].
\]

For a fixed FDR threshold $\alpha$, \citet{BH95} proposed a linear step-up FDR controlling procedure (the BH procedure) which sets the $p$-value cut off at $p_{(k)}$, where $k = \text{max}\{i : p_{(i)} \leq i \alpha / m \}$.   The procedure has been shown to control the FDR conservatively at level $\pi_0\alpha$ under independence, where $\pi_0 = m_0/m$ is the proportion of true nulls \citep{BY01}.   To tighten the FDR control, we could use the \emph{adaptive procedure} that applies the BH procedure at the threshold of $\alpha/\hat{\pi}_0$, where $\hat{\pi}_0$ is preferably a conservative estimate of $\pi_0$.

Instead of finding a rejection region to control the FDR, \citet{St02} proposed to estimate the FDR for a fixed rejection region.   When $R(t) > 0$ and under the usual assumptions that true null $p$-values are independent and uniformly distributed on $(0,1)$, a natural estimator for $\FDR(t)$ arises as
\[
	\widehat{\FDR}(t) = \frac{\hat{E}[V(t)]}{R(t)} = \frac{m \hat{\pi}_0 t}{R(t)}.
\]
The FDR control and FDR estimation approaches are intricately connected.   With $\hat{\pi}_0=1$, the BH procedure can be viewed as finding the largest $p$-value whose FDR estimate is below or equal to $\alpha$.

For a fixed tuning parameter $\lambda \in [0,1)$, \citet{St02} proposed a widely used $\pi_0$-estimator as
\[
\hat{\pi}_0^{}(\lambda) = \frac{m-R(\lambda)}{(1-\lambda)m}. \label{eq:pi0}
\]
Using $\hat{\pi}_0^{}(\lambda)$ in $\widehat{\FDR}$ leads to 
\[
\widehat{\FDR}_{\lambda}(t) = \frac{m \hat{\pi}_0(\lambda) t}{R(t) \vee 1},
\]
and \citet{LN12} showed that $\widehat{\FDR}_{\lambda}(t)$ is a conservative estimator of $\widehat{\FDR}(t)$, i.e.,
\[ E [\widehat{\FDR}_{\lambda}(t)] \geq \FDR(t).\]
To control the FDR in the adaptive procedure, it is a good practice to bound $\hat{\pi}_0$ away from zero, and Storey et al. \cite{St04} proposed an asymptotically equivalent estimator:
\[
\hat{\pi}_0^*(\lambda) = \frac{m-R(\lambda)+1}{(1-\lambda)m}.
\]
Because $\hat{\pi}_0^{*}(\lambda) \geq \hat{\pi}_0^{}(\lambda)$, using $\hat{\pi}_0^{*}(\lambda)$ in $\widehat{\FDR}$ leads to conservative estimation of the FDR.   On the other hand, Storey et al. \cite{St04} showed that the adaptive procedure with $\hat{\pi}_0^{*}(\lambda)$ controls the FDR.   Furthermore, if we use $\pi_0$-estimators that are more conservative than $\hat{\pi}_0^{*}(\lambda)$ in the adaptive procedures, the FDR control can also be guaranteed \citep{LN12}.   Such examples include the two-stage procedure of Benjamini et al. \cite{BKY06} and the one-stage and two-stage procedures of \citet{BR09}.   We will refer to the adaptive procedures that use fixed $\lambda$ parameters as \emph{fixed adaptive} procedures.   In summary, it is well established in the literature that for fixed adaptive procedures, conservative FDR estimation and FDR control are closely related.

In practice, the selection of $\lambda$ amounts to a trade-off between the bias and variance of $\hat{\pi}^*_0(\lambda)$ and should depend on the data at hand.   We will refer to the adaptive procedures that use data to select $\lambda$ as the \emph{dynamic adaptive} procedures.   Interestingly, \citet{LN12} showed that if $\lambda$ is chosen according to a certain stopping time rule, then conservative $\pi_0$ and FDR estimation can still be guaranteed.   Examples include the lowest-slope procedure of \citet{BH00} and the right-boundary procedure of \citet{LN12}. In spite of their conservative estimation, it is unclear whether such procedures will still maintain FDR control. Recently, \citet{HJ16} have proposed a class of weighted Storey $\pi_0$-estimators with data-dependent weights and showed that the corresponding dynamic adaptive procedures control the FDR in finite samples. However, the weight measurability condition required by \citet{HJ16} is not compatible with the stopping time condition required in the lowest-slope and right-boundary procedures, for which a proof of FDR control remains elusive.    

In this paper, we strive to prove the FDR control for a large class of dynamic adaptive procedures, which include the right-boundary and lowest-slope procedures as special cases.  The lowest-slope procedure is historically important in the field of multiple testing and especially in the FDR literature.   The lowest-slope $\pi_0$-estimator was first proposed in \citet{Hochberg90} to control familywise error rate (FWER), and its idea can be traced back to \citet{Schweder82}.   According to \citet{Benjamini10}, Benjamini and Hochberg attempted but could not show that the least-slope procedure controls the FDR and presented the non-adaptive BH procedure in \citet{BH95} as a result.   As the earliest adaptive FDR procedure, the lowest-slope procedure is widely used, but its control of the FDR has not been theoretically established.

We organize the rest of the paper as follows.   In Section~\ref{s2}, we show finite sample FDR control for a very general class of dynamic adaptive procedures and give specific examples of possible $\lambda$ selection rules.   In Section~\ref{s3}, we conduct simulation studies to demonstrate the advantages of dynamic adaptive procedures. Finally, we discuss the issues of identifiability, dependence, and discrete $p$-values, and conclude Section~\ref{s4}.   Technical proofs are postponed until Appendix \ref{app}.

\section{Dynamic adaptive procedures} \label{s2}

Throughout this paper, we will assume that the true null $p$-values are independent and identically distributed as $\text{Unif}(0,1)$ random variables, and are independent of the false null $p$-values. Under this model, arbitrary dependence is allowed among the false null $p$-values. This is the same condition adopted by \citet{BH95}, Storey et al. \cite{St04}, and \citet{LN12}, who call it the {\it null independence model}. Notice that under this model, the number of true nulls $m_0$ is fixed. A more general model with possibly random $m_0$ is termed as the {\it basic independence model} by \citet{HJ15}; note that results in the fixed $m_0$ model can be easily extended to the random case by conditioning on $m_0$ and integrating. We begin by presenting our main theoretical result.

\subsection{FDR control} \label{s21}

In this section, we will show the control of the FDR for the same class of stopping time rules used by \citet{LN12} to establish the conservative FDR estimation. Similar as in \citet{HJ15}, for $0 < \kappa < 1$, we divide the unit interval into a rejection region $[0, \kappa]$ and an estimation region $[\kappa, 1]$. We will first use the $p$-values in the estimation region to determine the tuning parameter $\lambda$ and the corresponding $\hat{\pi}_0^{*} (\lambda)$, then we decide the $p$-value rejection threshold in the rejection region. It may appear restrictive to limit the rejection threshold to be no greater than $\kappa$. In practice, we can set $\kappa$ not too small, say $\kappa=\alpha$, and it will be unlikely that the above restriction will affect the final rejection threshold. We refer readers to Remark 1 of Storey et al. \cite{St04} for a more detailed justification. 

We require the definition of the (forward) $p$-value filtration $\{\mathcal{F}_t\}_{t \in [\kappa, 1)}$, where $\mathcal{F}_t = \sigma( R(s) : \kappa < s \leq t)$. The $\sigma$-algebra $\mathcal{F}_t$ can be thought of as the information given by all the $p$-values located in the interval $(\kappa,t]$ plus the number of $p$-values no larger than $\kappa$.  The $\lambda$ selection rules considered are those such that $\lambda$ is a stopping time with respect to $\{\mathcal{F}_t\}_{t \in [\kappa,1)}$.

We define the following FDR estimator
\begin{equation} \label{fdr-star}
\widehat{\FDR}_{\lambda}^{*}(t) =   \left\{
\begin{array}{ll}
      \frac{m \hat{\pi}_0^{*} (\lambda) t}{R(t) \vee 1} & t \leq \kappa ,\\
      1 & t > \kappa .\\
\end{array} 
\right.
\end{equation}
Furthermore, for any function $F:[0,1] \rightarrow \mathbb{R}$, define the $\alpha$-level thresholding functional by
\[
t_{\alpha}(F) = \text{sup} \{ 0 \leq t \leq 1 : F(t) \leq \alpha \}.
\]
Then $t_{\alpha}(\widehat{\FDR}_{\lambda}^{*})$ is the rejection threshold for the dynamic adaptive procedure based on $\lambda$.

As our main theoretical result, we show that the dynamic adaptive procedures control the FDR.

\begin{thm} \label{t1}
Under the null independence model, suppose $\lambda$ is a stopping time with respect to $\{\mathcal{F}_t\}_{t \in [\kappa,1)}$, and satisfies $0 < \kappa \leq \lambda < 1$ almost surely for a fixed constant $\kappa$. Then 
\begin{equation*}
\FDR\{t_{\alpha}(\widehat{\FDR}_{\lambda}^{*})\} \leq \alpha.    
\end{equation*}
\end{thm}

The proof of Theorem~\ref{t1} and its required lemmas are presented in the Appendix~\ref{app}. Briefly, the proof of Theorem~\ref{t1} relies heavily on Lemma~\ref{l1}, whose proof follows that of Proposition 1 of \citet{HJ16}. We then construct a supermartingale (Lemma~\ref{l3}) and invoke the optional stopping theorem to bound the FDR below $\alpha$. 

The stopping time rules required in Theorem~\ref{t1} form a very general class, but it is not clear how they should be constructed in practice. For illustration purpose, we will analyze existing stopping time rules in the literature and show that they can be easily modified to satisfy the conditions of Theorem~\ref{t1}. Through this analysis, we will also draw insight and motivate new rules.

\subsection{Histogram-based $\lambda$ selection rules} \label{s22}

As an example of the stopping time $\lambda$ selection rule, we begin by formally defining the \emph{right-boundary} procedure \citep{LN12}.   For $k \geq 1$, consider a fixed and finite $\lambda$ candidate set $\Lambda=\{\lambda_1, \ldots, \lambda_k\}$ that divides the interval (0, 1] into $k+1$ bins with boundaries at $\lambda_0 \equiv 0 < \lambda_1 < \ldots < \lambda_{k} < \lambda_{k+1} \equiv 1$ such that the $i$th bin is $(\lambda_{i-1}, \lambda_{i}]$ for $i=1,\ldots,k+1$.   This partition resembles the construction of a histogram of the $p$-values.   Then the right-boundary procedure chooses the tuning parameter $\lambda = \lambda_j$, where
\begin{equation}\label{RB-stopping}
j = \min\{1 \le i \le k: \hat{\pi}^*_0(\lambda_i) \ge \hat{\pi}^*_0(\lambda_{i-1})\}
\end{equation}
if this set is non-empty, and otherwise chooses $j=k$.  That is, we choose $\lambda$ as the right boundary of the first bin where the $\pi_0$ estimate at its right boundary is larger or equal to that at its left boundary.   To ensure $\lambda \geq \kappa$, we can simply set $\lambda_1=\kappa$, or we can require that $\lambda_i \geq \kappa$ in addition to the condition $\hat{\pi}^*_0(\lambda_i) \ge \hat{\pi}^*_0(\lambda_{i-1})$ in (\ref{RB-stopping}). Then, it is clear that $\lambda$ is a stopping time with respect to $\{\mathcal{F}_t\}_{t \in [\kappa,1)}$. While in this definition we define a stopping rule using the more conservative estimator $\hat{\pi}_0^*(\lambda)$, $\lambda$ is still a stopping time if we substitute $\hat{\pi}_0(\lambda)$ in (\ref{RB-stopping}), as in the original right boundary procedure of \cite{LN12}. Such a substitution will only affect the $\lambda$ selection rule, and Theorem~\ref{t1} will still give finite sample control of the FDR as long as the thresholding procedure uses the estimator $\widehat{\FDR}_{\lambda}^{*}(t)$ defined in (\ref{fdr-star}).

It is straightforward to show that the right-boundary procedure chooses the first bin whose $p$-value density is less or equal to its tail average.   Typically, the overall $p$-value density shows a decreasing trend, and we want to choose a $\lambda$ not too small (large) to avoid high bias (variance). By design, the right-boundary procedure is likely to stop at a bin when the expected reduction in bias is comparable to the variation of $\hat{\pi}^*_0(\lambda)$.   In summary, the main idea behind the right-boundary procedure is to identify a $\lambda$ that would balance the bias and variance of the corresponding $\pi_0$-estimator.

The smaller the number of $\lambda$ candidates, the less sensitive the right-boundary procedure is to the change in $p$-value density. In the extreme case, if $k=1$, then the right-boundary procedure reduces to choosing a fixed $\lambda=\lambda_1$.   On the other hand, if we set $k$ too large, then we risk stopping too early and choosing a small $\lambda$ and its associated high positive bias in $\pi_0$ estimation.   This is because at each $\lambda$ candidate, there is a positive probability the procedure could stop, and checking the stopping condition too frequently will likely lead to early stop.   Past simulation studies \cite{LN12, Nettleton06} suggest that an equal-distance 20-bin setup is a reasonable choice for the number of tests $m$ in the thousands.

\subsection{Quantile-based $\lambda$ selection rules} \label{s23}

The histogram-based rules require the explicit specification of the $\lambda$ candidates beforehand.   Alternatively, we can let the data to determine the candidates by choosing $\lambda$ among $p$-value quantiles.   For example, Benjamini et al.\ \cite{BKY06} proposed the $k$-quantile procedure which selects $\lambda = p_{(k)}$ for some prespecified $1 \leq k \leq m$.   They recommended that $k=\lfloor\frac{m}{2}\rfloor$ such that $\lambda$ is approximately the median of the $p$-values.

\citet{BH00} proposed the lowest-slope procedure to control the FDR, which chooses the tuning parameter $\lambda = p_{(j)}$, where
\begin{equation}\label{LSL-stopping}
j = \min\{2 \le i \le m:  \hat{\pi}_0^*(p_{(i)}) > \hat{\pi}_0^*(p_{(i-1)})\}.
\end{equation}
Comparing to (\ref{RB-stopping}), it is easy to see that the right-boundary and lowest-slope procedures are essentially identical except that they use different $\lambda$ candidate sets.   Similar as our minor modification to the right-boundary procedure so that $\lambda \geq \kappa$, we can require that $p_{(i)} \geq \kappa$ in addition to the condition $\hat{\pi}_0^*(p_{(i)}) > \hat{\pi}_0^*(p_{(i-1)})$ in (\ref{LSL-stopping}).

Because the lowest-slope procedure checks its stopping condition at every realized $p$-value, it tends to stop too early and suffer high positive bias in estimating $\pi_0$. Not surprisingly, simulation studies in the literature have shown that the lowest-slope procedure is one of the most conservative and least powerful adaptive procedures, for example, see \cite{LN12} and \cite{Nettleton06}, among many others. In a sense, the lowest-slope procedure is penalized by the same multiplicity it tries to address. This undesirable result can be easily remedied by considering fewer stopping points, similar in spirit to choosing a reasonable number of bins for the right-boundary procedure.   We analogously define a {\em right-boundary quantile\/} procedure which applies the original right-boundary procedure to an arbitrary grid of fixed quantiles of the $p$-value distribution.    In the simulations to follow we will show that the right-boundary and right-boundary quantile procedures provide the best performance among dynamic adaptive procedures known to control the FDR in finite samples.

\section{Simulations} \label{s3}

We conducted simulation studies to evaluate the FDR control, power and $m_0$ estimation properties of the dynamic adaptive procedures in the literature.   The candidate procedures are

\begin{itemize}[label={--}]
\item \texttt{BH}, the original step-up procedure of \citet{BH95};
\item \texttt{ORC}, the oracle procedure by applying the BH procedure at level $\alpha /\pi_0$, assuming known $\pi_0$;
\item \texttt{RB20}, the right-boundary procedure with $\Lambda = \{0.05,0.1,...,0.95\}$;
\item \texttt{LSL}, the modified lowest-slope procedure;
\item \texttt{RB20q}, the right-boundary procedure that considers the 20 evenly spaced $p$-value quantiles. More specifically, $\Lambda = \{q_{0.05},q_{0.1},...,q_{0.95}\}$ where $q_\gamma$ denote the $\gamma$ quantile of the $p$-values;
\item \texttt{HJW}, the weight shifting method of \citet{HJ16}.
\end{itemize}
The simulation settings are similar to those in \citet{LN12}.   When true null $p$-values are independent, all the procedures considered are established to control the FDR in finite samples at level $\alpha$. \texttt{BH} controls the FDR conservatively at level $\pi_0 \alpha$. The finite sample control of \texttt{RB20}, \texttt{LSL} and \texttt{RB20q} all follow from Theorem~\ref{t1}. \texttt{HJW} is a particular example from a class of dynamic adaptive procedures shown to control the FDR\citep{HJ16}. We implement \texttt{HJW} as described in Section 5 of \cite{HJ16}, with fixed tuning parameters $\epsilon=0.05$, $k=6$, and $(\lambda_i)_{i=0}^k = (0.5,0.6,0.7,0.8,0.9,0.95)$.

Simulations are based on $J=10000$ replications, and the nominal FDR level is $\alpha=0.05$. For each replication, $m=10000$ one-sided tests of $H_0: \mu =0$ are performed, with standard normal true null statistics, and false null statistics having $N(\mu,1)$ distribution. Effect sizes $\mu$ are set to $0.5, 1, 2$ and 4. For effect sizes larger than 4, the false null $p$-values are well separated from the true null $p$-values, and all procedures achieve full power relative to \texttt{ORC}.

\subsection{Independent tests}

Simulation results for independent test statistics are reported in Figure~\ref{fig1}. The first row plots average realized FDR, the second the power relative to \texttt{ORC}, and the third the log mean-squared error (MSE) of $\hat{m}_{0,j} = \hat{\pi}_{0,j} m$, where $\{\hat{\pi}_{0,j}\}_{j=1}^J$ are the $\pi_0$ estimates from each of the $J$ replications, and MSE is defined as
\[
\text{MSE} = \frac{1}{J} \sum_{j=1}^J ( \hat{m}_{0,j} - m_0)^2.
\]
All procedures control the FDR below the nominal level 0.05 and see an increase in the FDR and relative power as the signal strength $\mu$ increases. \texttt{RB20} and \texttt{RB20q} provide the greatest relative power in all settings, and this is because they have the smallest MSE of $\hat{m}_0$. When the signal strength is larger, and the optimal $\lambda$ may be smaller than $\lambda_1 = 0.05$, the minimal possible value from \texttt{RB20}, in which case the quantile-based bins of \texttt{RB20q} can provide a marginal improvement over \texttt{RB20} by considering smaller stopping points, similar to the RB20* procedure in \citet{LN12}. \texttt{HJW}, although similar in spirit to \texttt{RB20}, cannot achieve the same power performance since it restricts its estimation region to $[0.5,1]$, and its right-to-left measurability condition forces it to sometimes over-weight the influence of smaller $p$-values in the estimation of $\pi_0$. Since it is known that \texttt{ORC} controls the FDR at exactly level $\alpha$ \cite{BY01}, all average realized FDR levels are corrected by the difference between the FDR of \texttt{ORC} and the target FDR level $\alpha$.

\begin{figure}[h!]
\includegraphics[width=\textwidth]{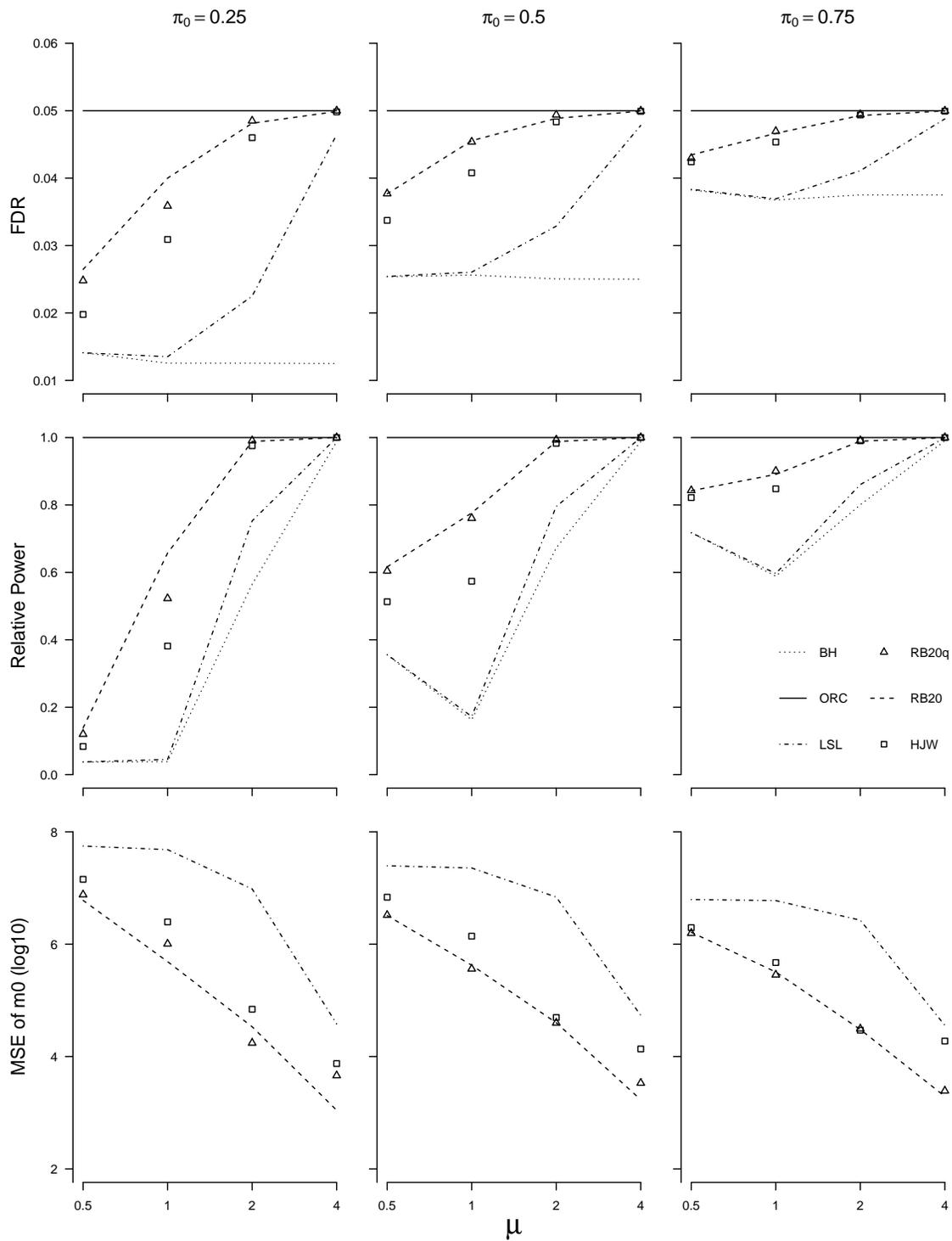}
\caption{Simulation results for independent test statistics.}
\label{fig1}
\end{figure}

\subsection{Dependent tests}

We also performed a simulation study with dependent test statistics. In particular, statistics have block auto-regressive order 1 correlation structure with block size 50 and correlation $\rho^{|i-j|}$  between the $i$th and $j$th elements in any block, and correlation coefficient $\rho = -0.9$. Block structure such as this has been used by \citet{LN12}, among others, to recreate the varying positive and negative correlations expected among genes in the same biological pathway. Results are reported in Figure~\ref{fig2}. As above, all procedures control the FDR below the nominal level 0.05 and increase in the FDR and relative power as the signal strength increases. \texttt{RB20} and \texttt{RB20q} remain the best in terms of power. There is some evidence that all procedures, including \texttt{ORC}, become conservative in the weak signal case, due to the dependence among the test statistics.

\begin{figure}[h!]
\includegraphics[width=\textwidth]{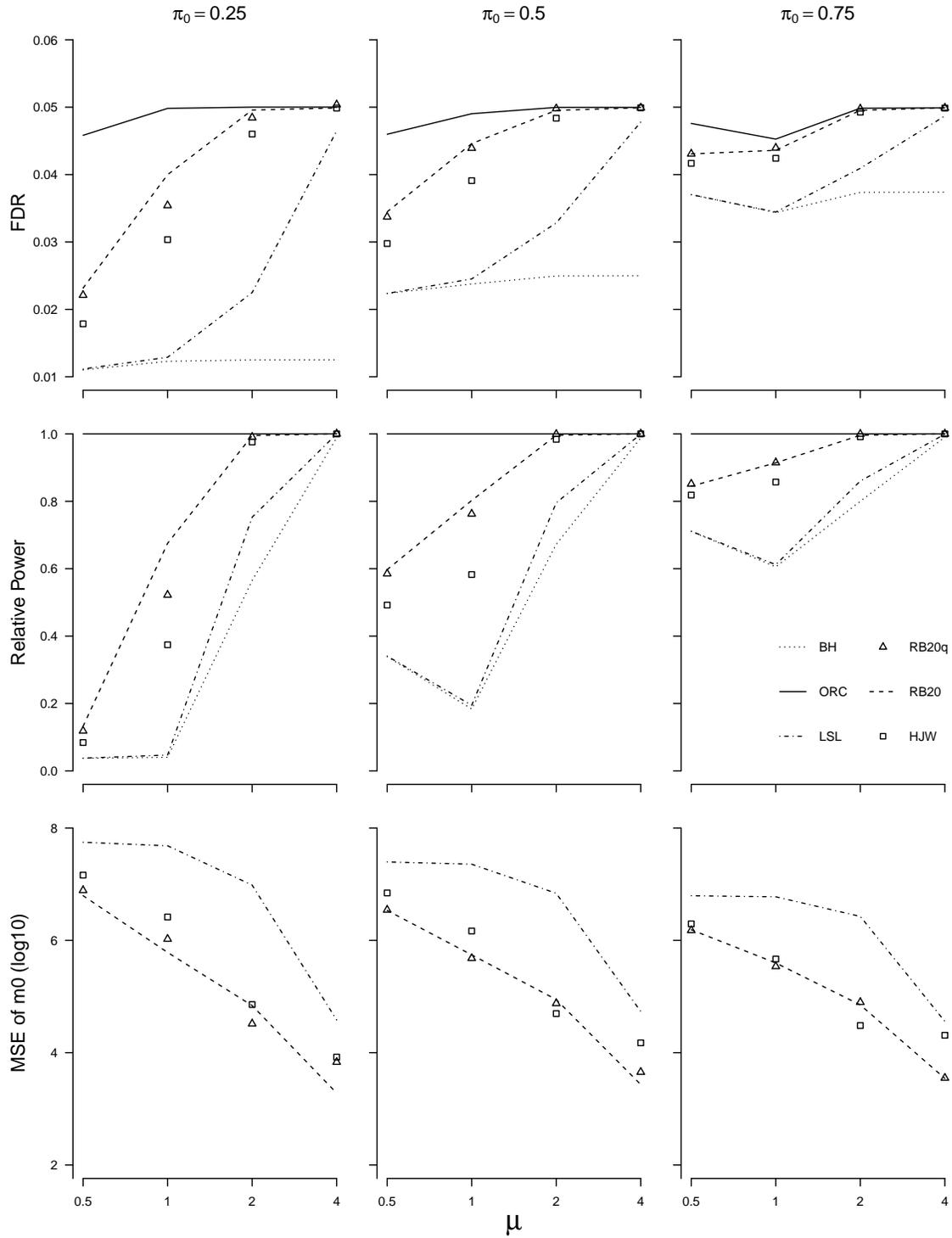}
\caption{Simulation results for correlated test statistics, $\rho=-0.9$.}
\label{fig2}
\end{figure}

\citet{LN12} show that, under weak dependence conditions, the dynamic adaptive procedures working with fixed grids (e.g., the right-boundary procedure) provide simultaneously conservative estimation and control of the FDR asymptotically.   For details, see Theorems 5 and 6 of \citet{LN12}.   Such theory explains why the FDR is under control in our dependent simulation setting. If we limit block sizes to be a constant or below a certain threshold and let the number of tests ($m$) increases, the weak dependence conditions are likely to hold.   This is because although tests within the same block are correlated, we will have more and more independent blocks as $m$ increases.

\section{Discussion and conclusions} \label{s4}

\subsection{Identifiability}

All of the results proven in this paper give only conservative control and estimation, rather than exact control or estimation. \citet{BY01}, among others, have shown that the original BH procedure has the FDR exactly equal to $\pi_0 \alpha$, but in the adaptive case in which we incorporate an estimate of $\pi_0$, identifiability issues manifest themselves, as discussed in Section 3.1 of \citet{GW04}. 

Under the commonly used {\em two-group model\/} \citep{Ef01} where all $p$-values are independent, with random $M_0 \sim BIN(m,\pi_0)$, and when $\lambda$ is selected using a stopping time rule from a fixed candidate set $\Lambda$, it can be shown that
\[
\FDR(t_{\alpha}(\widehat{\FDR}_{\lambda}^{*})) \leq \alpha \cdot \sup_{\lambda \in \Lambda} P(i \text{th null hypothesis is true } | ~ p_i > \lambda).
\]
There may in fact be no $\lambda$ for which $P(i \text{th null hypothesis is true } | ~ p_i > \lambda)=1$, a result of the false null $p$-value distribution having a non-zero uniform component. This is termed {\em impurity\/} by \citet{GW04}. Such purity issues are the reason that we cannot, without further assumptions on $F_1$, find an unbiased Storey-type estimator for $\pi_0$, and can only conclude conservatism.

\subsection{Dependence}

The results of this paper are proven under the classical null independence model, but prior FDR control literature has considered estimation and control properties under dependence assumptions on the true null $p$-values, in particular, the {\em positive regression dependence on a subset\/} (PRDS) condition in \citet{BY01} and the {\em reverse martingale dependence\/} (RMD) condition in \citet{HJ15}. 

Proposition 6.2 of \citet{HJ15} implies that finite sample control will not hold under every PRDS or RMD model, even for fixed adaptive procedures like those described by Storey et al.\ \cite{St04}. Nonetheless, it may still be possible to further limit the class of models or describe an alternative dependence model such that finite sample control can be proven for adaptive or even dynamic adaptive procedures. In particular, simulation studies in this paper and \citet{LN12} motivate that finite sample control may hold under certain types of block or autoregressive dependence. 

\subsection{Discrete $p$-values}

As most papers in the FDR literature, we have assumed that the true null $p$-values follow Unif(0,1). In many practical applications, discrete $p$-values are observed, and we will discuss the homogeneous and heterogeneous discrete cases separately. 

In many applications, the discrete $p$-values have a set of common support points, and we call such setting as the homogeneous discrete $p$-values setting.
For example, in high-throughput genetic experiments, $p$-values obtained through permutation tests have a set of identical support points.
The common support provides natural grid candidates, and \citet{Liang16} proposed the discrete right-boundary procedure that applies the idea of the right-boundary procedure to this setting and showed its conservative $\pi_0$ and FDR estimation.

For heterogeneous discrete $p$-values setting, which is also common in practice, many methods have been developed, but few have been shown to be powerful and control the FDR in finite samples. Recently, \citet{Dohler18} proved that several new procedures control the FDR and demonstrated their power in simulation studies. Adaptive procedures that incorporate $\pi_0$ estimates are suggested, but their control of the FDR has not been established and will be interesting future work.

\subsection{Conclusions}

For the adaptive procedures with $\hat{\pi}^*_0(\lambda)$ estimator, we show that if $\lambda$ is a forward stopping time, then the FDR is controlled. We demonstrated through simulation that the right-boundary procedure (\texttt{RB20}) and quantile-based right-boundary procedure (\texttt{RB20q}) outperform the competing dynamic adaptive procedures in terms of power and estimation accuracy of $\pi_0$ while maintaining FDR control at the nominal level.   In similar simulation settings in \citet{LN12}, the right-boundary procedure, with estimator $\hat{\pi}_0(\lambda)$ and a different candidate set $\Lambda$ than \texttt{RB20}, was shown to be more powerful than many other adaptive procedures, such as $\lambda=0.5$, the median adaptive procedure and the two-stage procedures of Benjamini et al.\ \cite{BKY06}, and the two-stage procedure of \citet{BR09}.   The simulation results thus far show that the right-boundary procedure is one of the most powerful adaptive procedures that controls the FDR.

Our results strengthen the connection between the FDR estimation approach and the FDR control approach.   With a conservative FDR estimator, we can use the step-up procedure to find the largest $p$-value whose FDR estimate is below the target FDR level and controls the FDR as a result. This connection is the most evident for fixed adaptive procedures through the work of Storey et al. \cite{St04} and \citet{LN12}. It is further studied for certain dynamic adaptive procedures by \citet{HJ15}. We extend this connection to still more dynamic adaptive procedures in this paper.   The FDR estimation approach is more direct, and conservative FDR estimation much easier to establish than finite sample control of the FDR.   Such insight could be useful in the future to design and evaluate new FDR controlling procedures.

\appendix

\section{Proofs} \label{app}

\subsection{Lemmas for Theorem 1}

We require the following three lemmas.

\begin{lem} \label{l1}

Under the conditions of Theorem 1, the dynamic adaptive procedure with the stopping time tuning parameter $1 > \lambda \geq \kappa > 0$ has
\begin{equation*}\label{Lemma1}
    \FDR\{t_{\alpha}(\widehat{\FDR}_{\lambda}^{*})\} \leq \frac{\alpha}{\kappa} E \bigg[ \frac{V(\kappa)}{m \cdot \hat{\pi}_0^*(\lambda)} \bigg].
\end{equation*}
\end{lem}

Lemma~\ref{l1} follows immediately from Proposition 1 of \citet{HJ16}. Although their Proposition 1 is established under the basic independence model, their proof still works under our null independence model.

Alternatively, Lemma~\ref{l1} can be viewed as a special case of Lemma 6.1 in \citet{HJ15} by noting that our stopping time condition satisfies their condition (A1) and the null independence model is a special case of their reverse martingale model.\\



\begin{lem} \label{l2}

Suppose $X \sim \text{BIN}(n,p)$. Then
\begin{equation*}
E \bigg[ \frac{1}{n-X+1} \bigg] \leq \frac{1}{(n+1)(1-p)}.    
\end{equation*}
\end{lem}

Lemma~\ref{l2} is given as Lemma 1 in \cite{BKY06}. 

\vspace{3mm}

Similar to the definition of the forward $p$-value filtration in Section~\ref{s21}, define the forward true null filtration $\{\mathcal{G}_t\}_{t \in [0,1]}$, where
\begin{equation*}
\mathcal{G}_t = \sigma(V(s) : 0 \leq s \leq t),    
\end{equation*}
and the $\sigma$-algebra generated by the false null $p$-values,
\begin{equation*}
\mathcal{S} = \sigma(S(t) : 0 \leq t \leq 1).
\end{equation*}

\begin{lem} \label{l3}

Define the filtration $\{\mathcal{H}_t\}_{t \in [0,1]}$ by
\begin{equation*}
\mathcal{H}_t = \sigma(\mathcal{G}_t, \mathcal{S}).
\end{equation*}
Then the process
\begin{equation*}
M(t) = \bigg\{ \frac{1-t}{m_0 - V(t) + 1} \bigg\}_{t \in [0,1]}
\end{equation*}
is a supermartingale with respect to $\{\mathcal{H}_t\}_{t \in [0,1]}$.
\end{lem}

\begin{proof}

Note that $0 \leq M(t) \leq 1$ a.s. for all $t \in [0,1]$, so certainly $M(t)$ is a collection of integrable random variables. Now fix $0 \leq s \leq t$. If $t=1$, then
\begin{equation*}
E \bigg[ M(t) \biggbar \mathcal{H}_s \bigg] = 0 \leq M(s).
\end{equation*}
Otherwise, $s \leq t < 1$, and
\begin{eqnarray*}
E \bigg[ M(t) \biggbar \mathcal{H}_s \bigg] &=& E \bigg[ M(t) \biggbar V(s) \bigg] \\
&=& E \bigg[ \frac{1-t}{m_0 - V(t) + 1} \biggbar V(s) \bigg] \\ 
&=& E\bigg[\frac{1-t}{(m_0-V(s)) - (V(t)-V(s)) + 1} \biggbar V(s) \bigg], \\ 
\end{eqnarray*}
where the first equality follows since $\mathcal{S}$ is independent of $M(t)$ by null independence. Also by null independence, conditional on $V(s)$,
\begin{equation*}
    V(t) - V(s) \sim \text{BIN}\bigg(m_0 - V(s), \frac{t-s}{1-s}\bigg).
\end{equation*}
Thus by Lemma~\ref{l2},
\begin{eqnarray*}
    E \bigg[ M(t) \biggbar \mathcal{H}_s \bigg] &\leq& \frac{1-t}{(m_0 - V(s) + 1)(1 - \frac{t-s}{1-s})} \\
    &=& \frac{1-s}{m_0 - V(s) + 1} \\
    &=& M(s).
\end{eqnarray*}

\end{proof}

\subsection{Proof of Theorem 1}

\begin{proof}

By Lemma~\ref{l1}, it follows that
\begin{equation*}
    \FDR\{t_{\alpha}(\widehat{\FDR}_{\lambda}^{*})\} \leq \alpha E \bigg[ \frac{V(\kappa)}{\kappa m \hat{\pi}_0^*(\lambda)} \bigg],
\end{equation*}
and hence it suffices to show that

\begin{equation*}
    E \bigg[ \frac{V(\kappa)}{\kappa m \hat{\pi}_0^*(\lambda)} \bigg] \leq 1.
\end{equation*}


Since $m - R(\lambda) \geq m_0 - V(\lambda)$,
\begin{eqnarray*}
    E \bigg[ \frac{V(\kappa)}{\kappa m \hat{\pi}_0^*(\lambda)} \bigg] &=& E \bigg[ \frac{1-\lambda}{m - R(\lambda) + 1} \cdot \frac{V(\kappa)}{\kappa} \bigg] \\
    &\leq& E \bigg[ \frac{1-\lambda}{m_0 - V(\lambda) + 1} \cdot \frac{V(\kappa)}{\kappa} \bigg]. \\
    &=& E\bigg\{ \frac{V(\kappa)}{\kappa} \cdot E \bigg[ \frac{1-\lambda}{m_0 - V(\lambda) + 1} \bigg| \mathcal{H}_{\kappa} \bigg] \bigg\} \\
    &\leq& E \bigg[ \frac{1-\kappa}{m_0 - V(\kappa) + 1} \cdot \frac{V(\kappa)}{\kappa} \bigg] \\
    &\leq& 1-\kappa^{m_0}\\
    &\leq& 1.
\end{eqnarray*}
The third to last inequality follows from Lemma~\ref{l3} and the optional stopping theorem \citep{Do53}. Note that $\lambda$ is a stopping time with respect to $\{\mathcal{F}_t\}_{t \in [\kappa,1)}$, and thus is also a stopping time with respect to the larger filtration $\{\mathcal{H}_t\}_{t \in [0,1]}$ to which the supermartingale in Lemma~\ref{l3} is adapted. The second to last inequality follows from the binomial argument of Storey et al. \cite{St04}, Theorem 3 (since $\kappa$ is a fixed constant).
\end{proof}

\bibliography{mybib}

\end{document}